\documentclass[a4paper,UKenglish]{lipics-v2021}

\nolinenumbers

\usepackage{algorithm2e}
\usepackage{float}

\newcommand{\T}{{\mathcal T}}

\renewcommand{\l}{{\ell}}

\newcommand{\close}[2]{close_{#2}({#1})}

\newcommand{\profit}[2]{profit_{#2}(#1)}

\newcommand{\sign}[1]{sgn({#1})}

\newcommand{\ml}[1]{#1}

\title{How Brokers can Optimally Abuse Traders}
\titlerunning{How Brokers can Optimally Abuse Traders} 

\author{Manuel Lafond}{Université de Sherbrooke, Canada}{manuel.lafond@USherbrooke.ca}{https://orcid.org/
0000-0002-5305-7372}{}

\authorrunning{Manuel Lafond}

\funding{The author acknowledges financial support from the Natural Sciences and Engineering Research Council (NSERC) and from the Fonds de Recherche du Québec – Nature et technologies (FRQNT) for this research.}


\Copyright{Manuel Lafond}

\EventEditors{Pierre Fraigniaud and Yushi Uno}
\EventNoEds{2}
\EventLongTitle{11th International Conference on Fun with Algorithms (FUN 2022)}
\EventShortTitle{FUN 2022}
\EventAcronym{FUN}
\EventYear{2022}
\EventDate{May 30--June 3, 2022}
\EventLocation{Island of Favignana, Sicily, Italy}
\EventLogo{}
\SeriesVolume{226}
\ArticleNo{17}

\begin{document}


\maketitle

\begin{abstract}
Traders buy and sell financial instruments in hopes of making profit, 
and brokers are responsible for the transaction.
There are several hypotheses and conspiracy theories arguing that in some situations, brokers want their traders to lose money.
For instance, a broker may want to protect the positions of a privileged customer.   
Another example is that some brokers take positions opposite to their traders', in which case they make money whenever their traders lose money.
These are reasons for which brokers might  manipulate prices in order to maximize the losses of their traders.

In this paper, our goal is to perform this shady task optimally --- or at least to check whether this can actually be done algorithmically.  Assuming total control over the price of an asset (ignoring the usual
aspects of finance such as market conditions, external influence or stochasticity),
we show how in quadratic time, given a set of trades specified by a stop-loss and a take-profit price, 
a broker can find a maximum loss price movement.
We also look at an online trade model where broker and trader exchange turns, each
trying to make a profit.  We show in which condition either side can make a profit, and that the best option for the trader is to never trade. 
\end{abstract}

\section{Introduction}

Trading is the practice of buying or selling financial assets with the aim of
making a profit. 
A trader can buy an instrument at some price $p$ and sell it at price $p'$, 
making a profit of $p' - p$ (which might be negative).  The trader can also sell an unowned instrument at some price $p$ with the obligation to buy it back someday, say at a time where  
the new price is $p'$, making a profit of $p - p'$ (this is called \emph{shorting}).  
A \emph{broker} is usually responsible for the execution of a trade, taking care of the technical aspects of the transaction.

This power over trade execution has given rise to several conspiracy theories.  This has been especially prevalent in the year 2021 where unprecedented financial events occurred.  One of them was the rise and fall of the GameStop (GME) stock that started in January.  In short, large institutions had significantly shorted the stock and, in a concerted counteract effort, retail traders massively bought it and made its price go up 30-fold.  At this point, several traders were unable to buy more GME stocks, the transaction being blocked by their broker (see~\cite{robinhood2021}).  This may have been caused by technical issues, but of course the Internet accused brokers of manipulating prices to protect the short positions of their large clients.  The stock went back near its original price, allowing the shorts to limit their losses, and then stabilized a bit higher in the months after.  
The year 2021 has also seen cryptocurrencies gain monstrous gains in value.  They are under less regulations than stocks and have been suspected of shady price manipulation techniques, for instance pumps and dumps, since their rise in popularity~\cite{pumpdump2021}.
As a final example, the foreign exchange currency market is largely managed by brokers called \emph{market-makers}.  These place  trades in the opposite direction of their traders --- if a trader wants to buy an asset, the market-maker will sell it, and if the trader
wants to sell it, they will buy it.  After all, there has to be two parties involved 
in a transaction, and market-makers assume one of the roles.    This gives them an incentive for their clients to perform poorly.  See~\cite{marketmakers2021} for a gentle introduction on market makers and~\cite{yang2011fx} for a thesis on the topic.
Price manipulation theories are based on the arguments that large brokerage firms have access to enough funds to control prices to some extent, along with statistics showing that a majority of traders lose money (see~\cite{losemoney2015}).

Now, taking such accusations seriously is likely to make economists jump out of their seats,
since there are too many factors driving asset prices for a single entity to control it.
But here, we rather take an algorithmic perspective on these theories.
To take it to the extreme, suppose that brokers have total control over the prices.
Armed with the knowledge of every trade that is currently open, 
their goal is to make people lose as much money as possible.  The question is: even with the ultimate power 
of price manipulation, \emph{can they}?  Is this optimization problem easy?
In this paper, we wear the hat of (would-be) mischievous brokers and 
devise price manipulation algorithms to do our evil bidding.
Note that this question has another interpretation: as a trader, what is the worst that could happen
with a set of opened trades\footnote{\ml{This suggests that traders could open trades in opposing directions.  This appears to make little sense, but it has its advantages --- in fact, this is a well-known risk management technique called \emph{hedging}}.}? 
There seems to be no algorithmic answer in the literature for this simple question.
Algorithms exist for the seemingly related notions of \emph{value at risk}~\cite{larsen2002algorithms} and \emph{maximum loss}~\cite{studer1997maximum},
but these measures are based on stochastic prices, operate on multiple assets and are 
subject to various market conditions, unlike here where we assume full control. 

\textbf{Our contributions.}  We model trades as bounded by two closing prices - a winning and a losing price.
This is typical in trading: traders want to limit their risk and often set up a \emph{stop-loss}, 
a price at which the trade closes automatically when too much losses are incurred. This is usually accompanied with a 
\emph{take-profit} price, which closes the trade when its profit is high enough.  Brokers can use 
these two pieces of information to their advantage, leading to \ml{several problems.
First, we study the \emph{offline} problem where, given a set of trades, we ask for a price movement that maximizes trader losses}.
We show that this problem reduces to finding a maximum independent set on the \emph{trade conflict graph}, which exhibits which trades cannot be won simultaneously.  This graph is bipartite and thus our problem 
can be solved in polynomial time using maximum flow techniques \cite{orlin2013max, stoer1997simple}.
By digging a bit deeper, we develop a specialized $O(n^2)$ time algorithm by characterizing trade conflict graphs geometrically.  That is, we define an equivalent graph class called \emph{bicolored plane domination graphs}, which are shown to be chordal bipartite (see~\cite{brandstadt1999graph}), and on which dynamic programming can be performed for our purposes.  This class is interesting in itself and leads to several open algorithmic and graph theoretical problems.

\ml{Second,} we look at the \emph{online} setting.  That is, the trader can add a new trade or close a trade
at any given time, and the broker still has to maximize losses.  This leads to a two-player
game where the trader and broker exchange turns.  We show that there is essentially one strategy that the broker must use, which is to always move the price greedily in the direction of maximum \emph{potential} profit.  In particular, using the independent set algorithm from above in the online setting incurs infinite losses against an optimal trader.  We conclude by showing that if the broker uses this strategy, the trader should simply not open any trade and get rich through other means.

\section{Preliminary notions}\label{sec:det-trades}

We use the notation $[a] = \{1,2, \ldots, a\}$ and $[a, b] = \{x \in \mathbb{Z} : a \leq x \leq b\}$.  Given an integer $x \in \mathbb{Z}$, the sign of $x$ is denoted $\sign{x}$, which is either $1$ or $-1$.  All graphs in this paper are finite and simple.  For a graph $G$ and $X \subseteq V(G)$, $G[X]$ denotes the subgraph induced by $X$.  A \emph{weighted graph} is a pair $(G, h)$ where $G$ is a graph and $h : V(G) \rightarrow \mathbb{R}$ assigns a weight to each vertex.  For $X \subseteq V(G)$, we write $h(X)$ for the sum of weights of vertices in $X$.
Some of our routines will require to sort integers.  We will write \emph{integer sorting time} to refer to the time required to sort a list of $n$ arbitrary integers, when $n$ is clear from the context (this can range from $O(n \log \log n )$ to $O(n)$ depending on conditions that we prefer to leave out of consideration~\cite{han2002deterministic, franceschini2007radix}).

Regarding the financial notions used in this paper, a \emph{price} is simply an integer, possibly negative.
Using integral prices is justified by the fact that in trading, prices are usually handled to the fourth of fifth decimal and may easily be treated  as integers.
All trades are done on a single asset, and for our purposes no fee is required to open a trade \ml{(although the fees could be embedded in the profit functions described below)}.
We assume that the current price of the asset is $0$.

A \emph{trade} $T = (w, \l, f)$ consists of three parameters:
$w$ is the winning price \ml{from the broker's point of view}, $\l$ is the losing price \ml{from the broker's point of view}, and $f : \{w, \l\} \rightarrow \mathbb{R}$ is the profit made by the broker if the price is attained before the other (see Figure~\ref{fig:main-example} and the description of price movements below). 
Note that profits can be negative.
We require that either $\l < 0 < w$, in which case $T$ is an $up$ trade, or $w < 0 < \l$, in which case $T$ is a $down$ trade.  
We will assume that $f(w) \geq f(\l)$ (otherwise, we can flip the roles of $w$ and $\l$).

Note that in a typical real-life trade setting, we have $f(w) > 0$ since $w$ represents a profit point and $f(\l) < 0$ since $\l$ it represents a loss point.
Moreover, profits and losses are usually linearly related, i.e. there is some $\delta$ such that $f(w) = f(\l) + \delta \cdot |w - \l|$ (as is the case in Figure~\ref{fig:main-example}).
However, such restrictions will not be needed for our algorithms until we discuss online trades, and we prefer to keep $f$ generic.
\ml{As another side note, in our terminology, an $up$ trade corresponds to a $sell$ from the trader since the mischievous broker makes profit when the trader loses, and likewise a $down$ trade is a $buy$ from the trader. Since we take the broker's viewpoint in this paper, trades are described with the broker's preferred direction.}

\begin{figure}[t]
  \begin{center}
    \includegraphics[width=\textwidth]{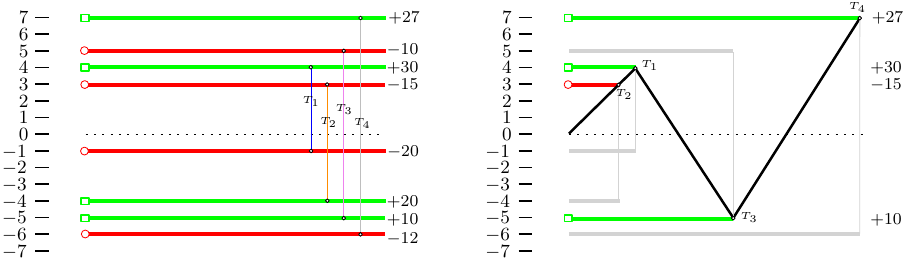}
  \end{center}
  \caption{Left: an example input for the Maximum Trader Abuse problem with $|\T| = \{T_1, T_2, T_3, T_4\}$.
  Each trade has a winning price (green lines that start with a square) 
  and a losing price (red lines that start with a circle).  Each winning and losing price has a corresponding profit indicated on the right.
  For instance, $T_1 = (4, -1, f)$ where $f(4) = 30$ and $f(-1) = -20$.
    (b) An optimal price movement $M$ for $\T$ that makes a total profit of $30 - 15 + 10 + 27 = 52$.  The green/red (square/circle) lines now indicate the lifespan of each trade according to $M$.  
    The profit points not realized by $M$ are grayed-out.  Only $T_2$ is lost since price $3$ is reached before $-4$.
     }
  \label{fig:main-example}
\end{figure}

A \emph{price movement}
$M = (m_1, \ldots, m_k)$ is a sequence
where $m_i \in \{+1, -1\}$ for each $i \in [k]$.
The current price after the $j$-th movement is $\sum_{i = 1}^{j}m_i$. 
Consider a trade $T = (w, \l, f)$ and a price movement $M$.  If, under $M$, price $w$ is reached before $\l$, then $T$ is \emph{won} and a profit of $f(w)$ is made, and if instead price $\l$ is reached before $w$, then $T$ is \emph{lost} and a profit of $f(\l)$ is incurred.  
The first price in $\{w, \l\}$ reached by $M$ is denoted $\close{T}{M}$ and is called the \emph{closing price} of $T$ under $M$.  If $M$ does not reach either $w$ or $\l$, then $\close{T}{M}$ is undefined.
If it is defined, then the broker's profit made from $T$ under price movement $M$ can be written as $f(\close{T}{M})$.  
Given a set of trades $\T$, we say that price movement $M$ is \emph{valid for $\T$} if $\close{T}{M}$ is defined for every $T \in \T$.
If $M$ is valid, the \emph{total profit} of $M$ for  $\T$ is 
\[
\profit{M}{\T} = \sum_{T = (w, \l, f) \in \T} f(\close{T}{M})
\]

The problem statement follows:

\vspace{3mm}

\noindent The \textbf{Maximum Trader Abuse} problem:\\
\noindent {\bf Given:} a set of trades $\T$; \\
\noindent {\bf Find:} a valid price movement $M$ that maximizes $\profit{M}{\T}$.

\vspace{3mm}

We will denote $n = |\T|$ unless stated otherwise.
We will say that $\T' \subseteq \T$ is an \emph{optimal set of trades} if there exists a price movement $M$ that maximizes $\profit{M}{\T}$ such that the set of trades won by $M$ is precisely $\T'$. \ml{This definition implies that}
the problem is equivalent to finding an optimal set of trades.

\ml{Observe that we could have removed the validity requirement on $M$ in the problem definition.  This would allow the broker to leave some trades open forever by never reaching one of their closing prices.  This appears to be an important difference, and we leave it as an open question whether our techniques can handle this variant.}
On another note, let us briefly discuss the size of $M$.
Consider a trade $T = (w, \l, f) \in \T$, and notice that $w$ or $ \l$ are not necessarily polynomial in $n$.
Since $M$ has to close $T$, representing $M$ as a sequence of unit movements might yield an exponential-size output,
which is not necessary.
Instead, $M$ can be described by the sequence of prices in which it changes direction, from which the +1/-1 sequence can easily be inferred.  
That is, we may represent the movement as a sequence of prices $M = (p_1, p_2, \ldots, p_k)$, where $p_1 = 0$ is the initial price.  Then for each $i \in \{0, \ldots, k-1\}$, $|p_{i+1} - p_i|$ steps must be performed in the same direction (up if $p_{i+1} > p_i$, and down otherwise).  In fact, each $p_i$ can be assumed to be one of the values in $\bigcup_{(w, \l, f) \in \T} \{w, \l\}$ since only these determine profits, which guarantees that the output can be made of linear size.

\section{Maximizing trader abuse in quadratic time}

We first show that finding an optimal $M$ reduces to finding a maximum weight independent set in a bipartite graph, which we call the \emph{trade conflict graph}, that describes which pairs of trades cannot be won together.  
The latter problem can be solved in polynomial time using standard maximum flow techniques.
We then study the trade conflict graph a bit more in-depth and show that its properties lead to a $O(n^2)$ time algorithm.  This is achieved by representing the trade as a set of green and red points on the 2D plane where conflicts correspond to bicolored domination relationships.

We say that a set of trades $\T$ is \emph{compatible} if there exists a price movement
$M$ such that for all $(w, \l, f) \in \T$, we have $\close{T}{M} = w$.  Otherwise $\T$ is \emph{incompatible}.  Two trades $T_1, T_2$ are compatible if the set $\{T_1, T_2\}$ is compatible (and otherwise incompatible).
Pairs of incompatible trades are easy to characterize.

\begin{lemma}\label{lem:pair-compatible}
Let $T_1 = (w_1, \l_1, f_1)$ and $T_2 = (w_2, \l_2, f_2)$ be two trades.  Then $T_1$ and $T_2$ are incompatible if and only if 
$\sign{w_1} \neq \sign{w_2}$, $|\l_1| \leq |w_2|$ and $|\l_2| \leq |w_1|$.
\end{lemma}

\begin{proof}
First assume that $w_1$ and $w_2$ have a different sign, and that $|\l_1| \leq |w_2|$ and $|\l_2| \leq |w_1|$ hold.  
As both trades are of opposite signs, it is impossible to reach price $w_2$ without reaching $\l_1$ first (or simultaneously if $w_2 = \l_1$), and it is impossible to reach $w_1$ without reaching $\l_2$ first (or simultaneously).
Thus, one of the loss prices must be hit, making the trades incompatible.

We prove the converse direction by contraposition.  Assume that one of the three conditions does not hold.  
Suppose first that $\sign{w_1} = \sign{w_2}$.
If $w_1, w_2 > 0$, then the price movement that always goes up wins both trades, and if $w_1, w_2 < 0$, we can make the price go down.  In either case, $T_1$ and $T_2$ are compatible.
We may thus assume that $\sign{w_1} \neq \sign{w_2}$.
If $|\l_1| > |w_2|$, we first win $T_2$ by moving the price to $w_2$.  This does not reach $\l_1$ and so $T_1$ is still open at this point.  It then suffices to go from $w_2$ to $w_1$, winning both trades.
The same idea applies when $|\l_2| > |w_1|$. 
\end{proof}

For example, trades $T_1$ and $T_2$ in Figure~\ref{fig:main-example} are incompatible.

Obviously, if a set of trades $\T$ has two trades $T_1, T_2$ that are incompatible,
then they cannot both be won and trivially $\T$ is incompatible.
Luckily, the converse also holds.

\begin{lemma}\label{lem:pairwise-compat}
Let $\T$ be a set of trades in which every pair of trades is compatible.
Then $\T$ is compatible.
Moreover, one can construct in integer sorting time a price movement $M$ that wins every trade of $\T$.
\end{lemma}

\begin{proof}
We use induction on $|\T|$.  The case $|\T| = 1$ is trivial: it suffices to go to the single winning price.
Now assume that $|\T| > 1$.  Note that for any proper subset $\T'$ of $\T$, each pair of $\T'$ is compatible.  We may thus assume by induction that the statement holds for any such $\T'$.
For a price $p$,
denote 
$\T(p) = \{(w, \l, f) \in \T : p = w$ or $p = \l\}$.
Let $p^+$ be the minimum value above $0$ such that $\T(p^+) \neq \emptyset$, and let $p^-$
be the maximum value below $0$ such that $\T(p^-) \neq \emptyset$.
If all trades of $\T(p^+)$ are won at price $p^+$, we can raise the price from $0$ to $p^+$, 
win these trades, bring the price back to $0$ and apply induction on $\T \setminus \T(p^+)$ to win every other trade.  
The same applies if all trades of $\T(p^-)$ are won at $p^-$.
So suppose that there are $T_1 = (w_1, \l_1, f_1) \in \T(p^+)$ and $T_2 = (w_2, \l_2, f_2) \in \T(p^-)$ that are losing, i.e. such that $\l_1 = p^+$ and $\l_2 = p^-$.
Then $T_1$ is a $down$ trade and $T_2$ is an $up$ trade, and hence $\sign{w_1} \neq \sign{w_2}$.
Moreover, $w_2 \geq \l_1 > 0$ by our choice of $p^+$, 
and similarly $w_1 \leq \l_2 < 0$.  These imply that $|\l_1| \leq |w_2|$ and $|\l_2| \leq |w_1|$ and, by Lemma~\ref{lem:pair-compatible}, 
$T_1$ and $T_2$ are incompatible, a contradiction.  It follows that some price movement can win every trade of $\T$.

One can derive a sorting time algorithm from the above argument. First obtain a sorted list $A_{+}$ of the all the trades $(w, \l, f)$ with respect to $\max \{w, \l\}$ in ascending order, and a sorted list $A_{-}$ of all the trades with respect to $\min \{w, \l\}$ in descending order. 
The first element of $A_{+}$ has price $p^+$, and the first price in $A_{-}$ is $p^-$.  
By advancing through $A_{+}$ (resp. $A_{-}$) as long as the trade prices are $p^+$ (resp. $p^-$), we build the set $L_{+}$ (resp. $L_{-}$) of the trades that are lost at price $p^+$ (resp. $p^{-}$).
By the above argument, one of $L_{+}$ or $L_{-}$ is empty.
If, say, $L_{+} = \emptyset$, we move the price to $p^+$, win every trade at that price, and add those trades to a set $C$ of closed trades. We also remove from $L_{-}$ those trades that just closed. 
We then move forward in $A_{+}$ to find the next $p^+$ and the next $L_{+}$ and, if necessary, we advance through $A_{-}$ and find the next $L_{-}$.   Again, one of $A_{+}$ or $A_{-}$ will be empty, so we repeat.  The process is the same if $L_{-}$ is empty instead.  Note that in subsequent iterations, one must check whether a trade is in $C$ before adding it to $L_{+}$ or $L_{-}$.
The whole procedure can be done by traversing $A_{+}$ and $A_{-}$ once, and by adding/removing each trade in $L_{+}$ or $L_{-}$ at most once, and $O(n)$ price movements are added, which takes overall $O(n)$ time after sorting.
\end{proof}

Lemma~\ref{lem:pairwise-compat} essentially states that we can devote our efforts
to finding a set of pairwise-compatible trades $\T' \subseteq \T$ 
that maximizes profit, although we need to account for the losses incurred by the trades \emph{not} in $\T'$.
We can thus reformulate the problem in graph-theoretical terms.

\begin{definition}
The \emph{trade conflict graph} $G(\T) = (V, E)$ of a set of trades $\T$ is the graph with vertex
set $V = \T$, and in which there is an edge between each pair of incompatible trades.  

Furthermore, the \emph{weighing of $\T$} is the function $h : \T \rightarrow \mathbb{R}$ that assigns the weight $f(w) - f(\l)$ to each trade $(w, \l, f) \in \T$.
\end{definition}

Note that all weights of $h$ are non-negative since $f(w) \geq f(\l)$ is assumed to hold for all trades.
Also observe that by Lemma~\ref{lem:pair-compatible}, testing compatibility on a pair of trades can be done in constant
time, and thus $G(\T)$ can be built in time $O(n^2)$.  Perhaps $G(\T)$ could be constructed faster with more refined ideas, but we shall not dwell on this.
Now, as per Lemma~\ref{lem:pairwise-compat}, $\T' \subseteq \T$ is compatible if and only if
$\T'$ forms an independent set in $G(\T)$ (a set of vertices with no shared edges).  As we show below, the $h$ weighing is adjusted so that our problem reduces to finding a maximum weight independent set.

\begin{theorem}
Let $\T$ be a set of trades.
Then $\T^* \subseteq \T$ is an optimal set of winning trades if and only if 
$\T^*$ is a maximum weight independent set in $(G(\T), h)$, where $h$ is the weighing function of $\T$.
\end{theorem}

\begin{proof}
Let $M$ be any price movement that closes every trade and let $\T'$ be the set of trades won by $M$.  
Note that $\T'$ must be an independent set of $G(\T)$. Moreover, the profit incurred by $M$ is
\begin{align*}
\sum_{T \in \T} \close{T}{M} &= \sum_{(w, \l, f) \in \T'} f(w) + \sum_{(w, \l, f) \in \T \setminus \T'} f(\l) \\
&= \sum_{(w, \l, f) \in \T'} f(w) + \left(\sum_{(w, \l, f) \in \T} f(\l) - \sum_{(w, \l, f) \in \T'} f(\l) \right) \\
&= \sum_{(w, \l, f) \in \T'} (f(w) - f(\l)) + \sum_{(w, \l, f) \in \T} f(\l) \\
&= \sum_{T \in \T'} h(T) + \sum_{(w, \l, f) \in \T} f(\l)
\end{align*}
The term $t := \sum_{(w, \l, f) \in \T} f(\l)$ does not depend on the choice of $M$ and can be ignored since it does not contribute to the optimization criterion.  
Therefore, each $M$ of profit $k + t$ corresponds to an independent set $\T'$ of weight $k$, and conversely each independent set of weight $k$ corresponds to a price movement of profit at least $k + t$, by Lemma~\ref{lem:pairwise-compat}.  It follows that finding an optimal $M$ is equivalent to finding an independent set in $G(\T)$ of maximum weight with respect to $h$.
\end{proof}

It is not hard to see that $G(\T)$ is bipartite.  Indeed, letting $\T^+ = \{(w, \l, f) \in \T : w > 0\}$ and $\T^- = \{(w, \l, f) \in \T : w < 0\}$, one can see by Lemma~\ref{lem:pair-compatible} that $\{\T^+, \T^-\}$ is a partition of $\T$ into two independent sets.
Finding a maximum weight independent set in a bipartite 
graph $G$ can be done in polynomial-time using a maximum-flow reduction.
This can be implemented to run in time $O(n^3)$ using the Stoer-Wagner min-cut algorithm~\cite{stoer1997simple}
or a variety of max-flow algorithms that run in time $O(nm)$, where $m$ is the number of conflicts~\cite{orlin2013max}.

\subsection{Bicolored plane domination models for trade conflict graphs}

We now reformulate trade conflict graphs in terms of two-colored points on the 2D plane.  
A \emph{colored point} is a triple $(x, y, c)$ where $x$ and $y$ refer to horizontal and vertical plane coordinates, respectively, and $c \in \{green, red\}$ is the color of the point.  
We say that a colored point $(x, y, c)$ \emph{dominates} another colored point $(x', y', c')$ if $x \geq x'$ and $y \geq y'$.

Let $G = (V, E)$ be a graph.
We say that $G$ is a \emph{bicolored plane domination graph} if
one can assign to each $v \in V$ a colored point $p_v = (x_v, y_v, c_v)$ such that $uv \in E(G)$ if and only if 
$c_u \neq c_v$ and the green point of $\{p_u, p_v\}$ dominates the red point of $\{p_u, p_v\}$.
In words, vertices correspond to colored points, and there is an edge between two vertices if their points are green and red, and if the green is to the upper-right of the red.
If this is the case, the multiset of colored points $p_u$ is called a \emph{bicolored plane domination model} for $G$ (multiset because two vertices could be assigned points with the same coordinates and color).
As it turns out, this coincides with trade conflict graphs.  See Figure~\ref{fig:bicolored} for an illustration.

\begin{figure}[t]
    \centering
    \includegraphics[width=\textwidth]{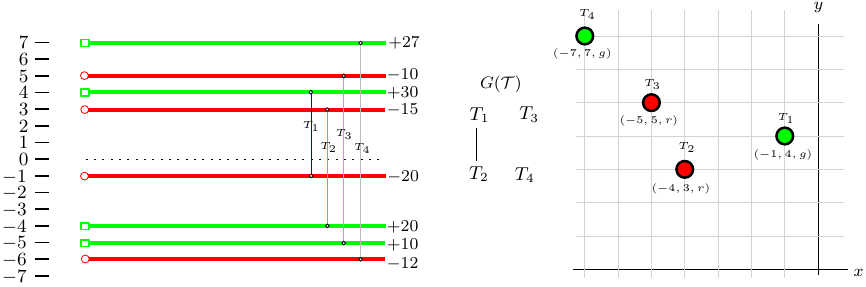}
    \caption{Left: the four trades from the previous example.  Middle: the graph $G(\T)$, which has only one incompatible pair.  Right: A bicolored plane domination model for $G(\T)$.  This is the model used in the proof of Theorem~\ref{thm:characterize}.}
    \label{fig:bicolored}
\end{figure}

\begin{theorem}\label{thm:characterize}
A graph $G$ is a trade conflict graph if and only if $G$ is a bicolored plane domination graph.

Moreover, given a set of trades $\T$, a bicolored plane domination model for $G(\T)$ can be found in time $O(n)$.
\end{theorem}

\begin{proof}
Let $\T$ be a set of trades and let $G := G(\T)$ be its trade conflict graph.  
Let $\T^+$ be the set of trades of $\T$ that are won at a positive price, and let $\T^- = \T \setminus \T^+$.
For each $T = (w, \l, f) \in \T$, assign the corresponding colored point of $T$ to
\begin{align*}
    (x_T, y_T, c_T) = \begin{cases}
       (\l, w, green) &\mbox{if $T \in \T^+$} \\
       (w, \l, red) &\mbox{if $T \in \T^-$} \\
    \end{cases}
\end{align*}
This model is illustrated in Figure~\ref{fig:bicolored}.
Now consider two distinct trades $T_1 = (w_1, \l_1, f_1)$ and $T_2 = (w_2, \l_2, f_2)$ of $\T$.  
First assume that $T_1$ and $T_2$ are incompatible.
Then by Lemma~\ref{lem:pair-compatible}, $\sign{w_1}  \neq \sign{w_2}$.  Assume without loss of generality that $T_1$ is in $\T^+$, and thus that $T_2$ is in $\T^-$.
By incompatiblity, $|\l_1| \leq |w_2|$.  Since $w_2 < 0$, it follows that $\l_1 \geq w_2$. 
Similarly, $|\l_2| \leq |w_1|$ and, since $w_1 > 0$, we have $w_1 \geq \l_2$.  Therefore, the point $(\l_1, w_1, green)$ corresponding to $T_1$ dominates the point $(w_2, \l_2, red)$
corresponding to $T_2$, and the model correctly puts an edge between $T_1$ and $T_2$.

Assume that $T_1$ and $T_2$ are compatible.  
If $\sign{w_1} = \sign{w_2}$, then the points corresponding to $T_1$ and $T_2$ are either both green and both red and the model correctly puts a non-edge between $T_1$ and $T_2$.
We may thus assume without loss of generality that $T_1 \in \T^+$ and $T_2 \in \T^-$.
By compatibility, one of $|\l_1| > |w_2|$ or $|\l_2| > |w_1|$ holds.
Suppose that $|\l_1| > |w_2|$.  Since $\l_1$ is negative, we have $\l_1 < w_2$, which means that the green point $(\l_1, w_1, green)$ does not dominate the red point $(w_2, \l_2, red)$.  So suppose that $|\l_2| > |w_1|$.  Since $\l_2$ is positive, we have $\l_2 > w_1$, and again $(\l_1, w_1, green)$ does not dominate $(w_2, \l_2, red)$.
In either case, $T_1$ and $T_2$ do not share an edge according to the model.
We deduce that $G$ is a bicolored plane domination graph using the model described above.  

Note that constructing the above set of colored points does not require constructing $G(\T)$.  It only requires looping through $\T$ and outputting each point in $O(1)$ time, which takes total time $O(n)$.  This proves the second part of the theorem statement.

Now consider the converse direction of the equivalence. Assume that $G$ is a bicolored plane domination graph.  We show that there exists a set of trades $\T$ such that $G$ is isomorphic to $G(\T)$.  
Take any model for $G$ and, for $v \in V(G)$, denote by $(x_v, y_v, c_v)$ the colored point assigned to $v$.
Notice that bicolored plane domination models are robust to translation.  That is, if we translate every point by the same vector, the domination relationships remain unchanged.  Using this, we can assume that every point is in the upper left quadrant of the plane, i.e. that every point $(x_v, y_v, c_v)$ satisfies $x_v < 0$ and $y_v > 0$ (which can be done by translating to the left and upwards by large enough amounts).  
For each $v \in V(G)$ such that $c_v = green$, add to $\T$ the trade $(y_v, x_v, f_v)$, where $f_v$ is arbitrary for our purposes.
Then for each $w \in V(G)$ such that $c_{w} = red$, add to $\T$ the trade $(x_{w}, y_{w}, f_w)$ again with arbitrary $f_w$.

We show that dominating pairs of green-red points coincide with incompatible trades.
Consider a green point $(x_v, y_v, green)$ that dominates a red point $(x_w, y_w, red)$. 
Then the corresponding trades $(y_v, x_v, f_v)$ and $(x_w, y_w, f_w)$ are incompatible because $\sign{y_v} \neq \sign{x_w}$ and $y_w \leq y_v$ by domination, which implies $|y_w| \leq |y_v|$ since the $y$'s are positive, and because $x_w \leq x_v$ by domination, which implies $|x_v| \leq |x_w|$ since the $x$'s are negative.
Next, consider two incompatible trades $T_1, T_2$ of $\T$. 
They must be of the form $T_1 = (y_v, x_v, f_v)$ and $T_2 = (x_w, y_w, f_w)$.  By our construction, the corresponding points are $(x_v, y_v, green)$ and $(x_w, y_w, red)$.  By incompatibility, $|x_v| \leq |x_w|$ and $|y_w| \leq |y_v|$.  Since all the $y$'s are positive and all the $x$'s are negative, this implies $x_w \leq x_v$ and $y_w \leq y_v$ and thus $(x_v, y_v, green)$ dominates $(x_w, y_w, red)$.
It follows that $G(\T)$ has exactly the same edges as $G$.
\end{proof}

\ml{The above equivalence arguably simplifies our trading framework, but it is not immediately convenient to design algorithms using a generic bicolored plane domination model.  We proceed to show that any such model can be transformed into an equivalent one on a discrete grid in which each column has exactly one point, and each row has exactly one point, without gaps.  This will allow us to design a relatively simple dynamic programming algorithm.}

We say that a multiset of $n$ colored points $P$ \emph{has the  permutation matrix property} if $\{x : (x, y, c) \in P\} = \{y : (x, y, c) \in P\} = [n]$.
In other words, each point of $P$ has a distinct $x$ coordinate in $[n]$, and a distinct $y$ coordinate in $[n]$.  The property is named after the fact that if an $n \times n$ matrix is filled with $0$s, but that we put at $1$ in each $(x, y)$ that occurs in $P$, then we have a permutation matrix.

\begin{lemma}\label{lem:makepermut}
Let $P$ be a multiset of $n$ colored points representing a bicolored plane domination model of some graph $G$.
Then in integer sorting time, one can construct from $P$ another bicolored plane domination model $P^*$ for $G$ such that $P^*$ has the permutation matrix property.
\end{lemma}

\begin{proof}
Suppose that $P$ does not already have the permutation matrix property.  We will assume that the points of $P$ are in the upper-right quadrant of the plane, i.e. that each $(x, y, c)$ satisfies $x > 0$ and $y > 0$.  This can be achieved through translation as we did previously.

We first argue that we can find another model $P'$  for $G$ in the same quadrant such that each $x$-coordinate is distinct and each $y$-coordinate is distinct.  First, multiply the $x$ coordinate of each point of $P$ by $n$, and notice that we obtain a model of the same graph since this does not alter domination relationships.  
Then, sort the points of $P$ in ascending lexicographic order, where $red$ comes before $green$.  That is, points are of the form $(x, y, c)$ and we sort by $x$ coordinates first, then by $y$, and then $c$ to break ties (again, note that if there are points that have the same $x$ and $y$ coordinates, those that are $red$ occur before the $green$).
Traverse the points using this order and suppose that, for some $x$, we encounter several points $(x, y_1, c_1), \ldots, (x, y_k, c_k)$ with the same $x$-coordinate, in this sorted order.
Replace them by the points $(x, y_1, c_1), (x + 1, y_2, c_2), \ldots, (x + k - 1, y_k, c_k)$.  Because we previously multiplied all coordinates by $n$, these points now all have an $x$-coordinate not shared with any other point.
Moreover, one can check that our sorting criteria ensure that domination relationships are unchanged (even for green-red pairs having the same $x$ and $y$ positions since the green is moved to the right of the red).
After traversing all the points in this fashion, all $x$ values are distinct.  Note that none of the $y$ values have changed in this process.  We can thus repeat the same idea with the $y$-coordinates, resulting in another model $P'$ for $G$ with no repeated $x$ nor $y$ values.  This requires sorting time, plus one pass through the points in time $O(n)$.

To get the permutation matrix property, it suffices to ``squish'' all the $x$'s and the $y$'s in $[n]$.
That is, sort $P'$ in ascending  $x$-coordinates and let $(x_1, y_1, c_1), \ldots, (x_n, y_n, c_n)$ be the obtained ordering (there are no ties this time).  Replace the points by \linebreak $(1, y_1, c_1), (2, y_2, c_2), \ldots, (n, y_n, c_n)$, keeping the colors the same, and note that again, the domination relationships are unchanged.  This takes integer sorting time, assuming we have to sort again.  We can repeat with the $y$'s and obtain an equivalent model with the permutation matrix property in total integer sorting time.
\end{proof}

We now proceed to finding a maximum weight independent set.  
Let $(G, h)$ be a weighted bicolored plane domination graph with a model $P$ that has the permutation matrix property.  For simplicity, we will denote the vertices of $G$ by their points, so that $V(G) = P$.
For $i, j \in [n]$, denote by 
\[
V(i, j) = \{(x, y, c) \in V(G) : 1 \leq x \leq i \mbox{ and } j \leq y \leq n\}
\]
In other words, $V(i, j)$ is the set of vertices that correspond to points located in the box with corners $(1, n)$ and $(i, j)$.
We will denote by $I(i, j)$ the maximum weight of an independent set of $G[V(i, j)]$ with respect to the $h$ weights.
The maximum weight of an independent set of $G$ is $I(n, 1)$, the value we are ultimately looking for.

Since $P$ has the permutation matrix property, for $j \in [n]$, we will denote by $(x_j, j, c_j)$ the unique point of $P$ whose $y$-coordinate is $j$.
One last notation: for $i_1, i_2, j \in [n]$, denote by $R(i_1, i_2, j) = \{(x, y, c) \in P : i_1 \leq x \leq i_2$ and $j \leq y \leq n$ and $c = red\}$.  See Figure~\ref{fig:dp-example} for an illustration.

\begin{figure}
    \centering
    \includegraphics{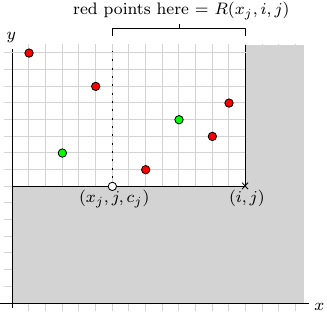}
    \caption{Illustration of the main components of the $I(i, j)$ recurrence.  Points in the grayed area are ignored and only the points of $V(i, j)$ are shown.  Here, $(x_j, i, c_j)$ is the unique point with $y$-coordinate $j$, which could be red or green.  The $R(x_j, i, j)$ set consists of the three red points in the area shown.}
    \label{fig:dp-example}
\end{figure}

For $i, j$ such that $i \notin [n]$ or $j \notin [n]$, the value of $I(i, j)$ is $0$.  
Otherwise, we can use the following recurrence:

\begin{align*}
    I(i, j) = 
    \begin{cases}
        I(i, j + 1) &\mbox{if $x_j > i$} \\
        h((x_j, j, c_j)) + I(i, j + 1) &\mbox{if $x_j \leq i$ and $c_j$ is green} \\
        \max (I(i, j + 1), 
            h(R(x_j, i, j)) + I(x_j - 1, j + 1))
        &\mbox{if $x_j \leq i$ and $c_j$ is red}
    \end{cases}
\end{align*}

\begin{lemma}
The above recurrence for $I(i, j)$ correctly represents the maximum weight of an independent set of $G[V(i, j)]$.
\end{lemma}

\begin{proof}
This can be shown inductively on $j$ in reverse order.
As a base case, notice that for $j > n$ and any $i$, then $I(i, j) = 0$ is correct since $V(i, j)$ is empty.
Assume that $i, j \in [n]$ and that for any $j' > j$ and any $i'$, the recurrence for $I(i', j')$ is correct.  
Let $W$ be an independent set of $G[V(i, j)]$ of maximum weight.
Consider the point $(x_j, j, c_j)$ of $P$.
If $x_j > i$, then point $(x_j, j, c_j)$ is not in $V(i, j)$ and $W$ does not contain it.  Thus all points of $W$ are in $V(i, j + 1)$ and $I(i, j) = I(i, j + 1) = h(W)$ follows by induction.
We may assume that $x_j \leq i$.  
Suppose that $(x_j, j, c_j)$ is green.  By the permutation matrix property, $(x_j, j, c_j)$ does not dominate any point in $V(i, j)$.  Hence, $(x_j, j, c_j)$ is an isolated vertex of $G[V(i, j)]$ and we may assume that $(x_j, j, c_j) \in W$ since all weights are positive.  All other points of $W$ must be in $V(i, j + 1)$ and by induction $h(W \setminus \{(x_j, j, c_j)\}) = I(i, j + 1)$.  If follows that $I(i, j) = h((x_j, j, c_j)) + I(i, j + 1) = h(W)$ is correct.

Finally, suppose that $(x_j, j, c_j)$ is red.  
\ml{We first show that $I(i, j) \geq h(W)$ and deal with the converse bound after.}
\ml{There are two cases: }either $(x_j, j, c_j) \in W$, or not.  
If not, then all points of $W$ are in $V(i, j + 1)$ and $h(W) = I(i, j + 1)$.
But if $(x_j, j, c_j) \in W$, then no green point $(x_k, k, c_k)$ of $V(i, j)$ such that $x_k > x_j$ can be in $W$, since any such point dominates $(x_j, j, c_j)$.  Therefore, we may assume that every \emph{red} point $(x_k, k, c_k)$ of $V(i, j)$ with $x_k > x_j$ is in $W$.  
These points are $R(x_j, i, j)$, and so $W$ consists of $R(x_j, i, j)$ plus a maximum weight independent set of $V(x_j - 1, j + 1)$ of weight $I(x_j - 1, j + 1)$.  
\ml{In either case}, since $I(i, j)$ takes the maximum of the two possibilities, we get that $I(i, j) \geq h(W)$.
For the converse bound, it is not hard to see that there exists an independent set of weight $I(i, j + 1)$ in $V(i, j)$, and an independent set of weight $h(R(x_j, i, j)) + I(x_j - 1, j + 1)$ in $V(i, j)$ (the latter obtained by taking all red points in $R(x_j, i, j)$ and an optimal independent set in $V(x_j - 1, j + 1)$, noticing that none of the points from that set can dominate $R(x_j, i, j)$).  The weight of $W$ is at least the maximum of those, and so $I(i, j) \leq h(W)$.
This concludes the proof.
\end{proof}

Algorithm~\ref{alg:mainalgo} summarizes the above ideas and computes a maximum weight independent set in a bicolored plane domination graph in $O(n^2)$ time, as we argue below.

\begin{algorithm}
\DontPrintSemicolon
\SetKwProg{Fn}{function}{}{}
\Fn{abuseTraders($\T$)}{
    Compute the bicolored plane domination model $P$ for $G(\T)$ (Theorem~\ref{thm:characterize})\;
    Compute a model $P^*$ with the permutation matrix property (Lemma~\ref{lem:makepermut})\;
    $computeRweights(P^*)$\;
    \For{$j = n$ down to $1$}
    {
        \For{$i = 1$ to $n$}
        {
            Compute $I(i, j)$ using the recurrence\;
        }
    }
    return $I(1, n)$\;
}
\Fn{computeRweights($P^*$)}
  {
    \For{$j = 1$ to $n$}
    {
      Let $(x_j, j, c_j)$ be the unique point of $P^*$ with $y$-coordinate $j$\;
      $weight = 0$\;
      \For{$i = x_j$ to $n$}
      {
         Let $(i, y_i, c_i)$ be the unique point of $P^*$ with $x$-coordinate $i$\;
         \uIf{$c_i = red$ and $y_i > j$}
         {
            $weight = weight + h((i, y_i, c_i))$\;
         }
         $h(R(x_j, i, j)) = weight$\;
      }
    }
  }

  \caption{How to abuse traders: a geometric interpretation. }
  \label{alg:mainalgo}
\end{algorithm}
\vspace{3mm}

The algorithm simply computes $I(i, j)$ for every $j \in [n]$ in reverse order and for every $i \in [n]$ in order.
There are $O(n^2)$ values of $I(i, j)$ to compute and each can take time $O(1)$, assuming that each $h(R(x_j, i, j))$ can be accessed in constant time.  
This can be achieved by storing all the relevant $R$ values during a preprocessing step.  Notice that the only relevant values to store for the recurrence have the form $R(x_j, i, j)$, where $x_j$ is determined by $j$.  Hence, for each $x_j$ and $j$ pair, we can compute $h(R(x_j, i, j))$ for each $i$ in increasing order of $i$, adding for each $i$ the weight of the new point with $x$-coordinate $i$ if relevant.  For each $j$, we need linear time to compute the necessary $h(R(x_j, i, j))$ values, so this preprocessing is done in time $O(n^2)$.
All the above ideas are shown in Algorithm~\ref{alg:mainalgo}.
We note that this algorithm only returns the maximum weight, but a standard backtracking procedure can find an actual set of optimal winning trades, which in turn can be converted to a price movement.
We have argued the following.

\begin{theorem}
A maximum weight independent set of a bicolored plane domination graph can be found in time $O(n^2)$.
\end{theorem}

\subsection{Some open questions on trade conflict graphs}

If we abstract away the financial motivations of this work for a moment, trade conflict graphs (i.e. bicolored plane domination graphs) can be studied from a purely graph theoretical perspective.  One could ask whether such graphs are easy to recognize, and whether trade conflict graphs coincide, or at least share some similarity with a graph class that is already known.
Although we reserve these questions for future work, we can provide a partial answer to the latter.
A graph is \emph{chordal bipartite} if it is bipartite and contains no induced cycle with at least $6$ vertices.

\begin{proposition}
All bicolored plane domination graphs are chordal bipartite.  
\end{proposition}

\begin{proof}
We already know that $G$ is bipartite, so we only need to argue that it has no induced cycle of length $6$ or more.
Let $G$ be a bicolored plane domination graph and let $P$ be a set of colored points that is a model for $G$.
Assume that $G$ contains an induced cycle $r_1 g_1 r_2 g_2 \ldots r_k g_k r_1$ with $k \geq 3$.  Suppose without loss of generality that the $r_i$'s correspond to red points and the $g_i$'s to green points in $P$ (otherwise, start the cycle at $g_1$ and rename accordingly).  We will denote by $r_i.x$ and $r_i.y$ the coordinates of the point corresponding to $r_i$ in $P$, and use the same notation $g_i.x, g_i.y$ for $g_i$, where $i \in [k]$.

Observe that for any $i \in [k]$, it is not possible that $r_i.x \geq r_{i+1}.x$ and $r_i.y \geq r_{i+1}.y$ both hold (where $r_{k+1}$ is taken to be $r_1$).
Indeed, this would imply that any green point that dominates $r_i$ would also dominate $r_{i+1}$, and we know that $g_{i-1}$ dominates $r_i$ but not $r_{i+1}$ (note that this is where we need the cycle to be of length at least $6$).  
By a symmetric argument, we cannot have $r_{i+1}.x \geq r_i.x$ and $r_{i+1}.y \geq r_i.y$ simultaneously.

\begin{figure}
    \centering
    \includegraphics[width=0.4\textwidth]{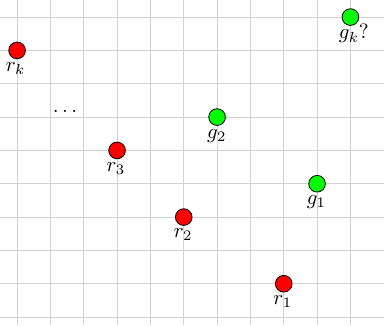}
    \caption{An illustration of how some of the relative locations of the $r_i$ and $g_i$ points are forced. }
    \label{fig:chordalbipartite}
\end{figure}

The rest of the following argument is illustrated in Figure~\ref{fig:chordalbipartite}.
Suppose without loss of generality that $r_2.x \leq r_1.x$ (otherwise, start the cycle at $r_2$ and list the vertices in reverse order by renaming accordingly).
We then deduce from the above observation that $r_2.y > r_1.y$.
Moreover, $g_2$ dominates $r_2$ but not $r_1$, which means that
$g_2.y \geq r_2.y > r_1.y$ but $g_2.x < r_1.x \leq g_1.x$.
Now consider the location of $r_3$, which is dominated by $g_2$.
We must have $r_3.x \leq g_2.x \leq g_1.x$ and, since $g_1$ does not dominate $r_3$, we must have $r_3.y > g_1.y \geq r_2.y$.
Again using our previous observation, we get that $r_3.x \leq r_2.x$.

In other words we have deduced from $r_2.x \leq r_1.x$ and $r_2.y > r_1.y$ that $r_3.x \leq r_2.x$ and $r_3.y > r_2.y$.
This argument can be used inductively to infer that 
$r_k.x \leq r_{k-1}.x \leq \ldots \leq r_1.x$ and $r_k.y > r_{k-1}.y > \ldots > r_1.y$.
Now consider the point $g_k$, which dominates both $r_k$ and $r_1$.  One can see that such a point must dominate each point in $r_1, r_2, \ldots, r_k$, a contradiction.
Thus the cycle cannot exist.
\end{proof}

Let us mention that chordal bipartite graphs are known to admit $O(|V| + |E|)$ time algorithms for the maximum \emph{unweighted} independent set problem, assuming that an ordering of the vertices called a \emph{strong ordering} is given~\cite{dragan2000strongly}.  However, it is not clear whether the same time bound can be achieved for the weighted version.  In any case, the geometric perspective may lead to even better algorithms in the future.
We close this section with some question that this leads us to.

\begin{enumerate}
    
    \item 
    Is the class of bicolored plane domination graphs equal to the class of chordal bipartite graphs?

    \item 
    If not, can bicolored plane domination graphs be characterized as bipartite graphs that forbid induced cycles of length at least $6$, plus a finite set of forbidden induced subgraphs?

    \item 
    Given a graph $G$, can one decide in polynomial time whether $G$ is a bicolored plane domination graph?
    
    \item 
    Given a bicolored plane domination model $P$ that has the permutation matrix property, can a maximum weight independent set be found in time $O(n)$ (regardless of the number of edges of the underlying graph)?
    
\end{enumerate}
 
A positive answer to (1) would be a purely accidental find, and would be surprising since bicolored  plane domination graphs appear to have a simple structure.  However, it seems hard to construct a chordal bipartite graphs that can be shown to not be a trade conflict graph.
This would probably give clues on (2) though.  
As for (3), trade conflict graphs are probably polynomial-time recognizable, given their simplicity, but this requires a more careful look.
As for (4), it is plausible that $o(n^2)$ could be achieved, since our recurrence needs to go through lots of $(i, j)$ locations that contain no point, and only $n$ such locations appear truly relevant.

\section{Online trades}\label{sec:online-trades}

We now allow the trader to interrupt the broker's price manipulation
at any time by adding a new trade or closing an open trade.  
One may ask whether, given this new power, the trader can make money.
It does appear that a small proportion of real-life traders makes profit, so perhaps the broker can be beaten.  On the other hand, if there was a strategy for the trader that guarantees profit 
even against conspiratorial prices, someone would have figured it out by now.
But surprisingly, there doesn't seem to be a clear answer in the literature.
One impossibility result is given by the \emph{efficient market hypothesis} (EMH)~\cite{fama1970efficient}, 
which, roughly speaking, states that if prices perfectly reflect their environment, then no strategy can
win \emph{consistently}, i.e. against every price movement (note that the author is far from being an economist and that the actual theory is much deeper).  
In the same vein, assuming the price is a random walk, the best trader strategy 
has expected profit $0$, so there exists a price movement with no win for the trader (unless the trader has infinite funds, in which case a Martingale betting system can be used for profit).  
Thus the broker can ``mimic'' a worst-case EMH-driven price
or a random walk to guarantee that, on expectation, traders make no money. 
The broker can still lose money if unlucky, raising the question of whether there is a strategy for the broker that guarantees profitability.  \ml{To our knowledge, none of aforementioned frameworks is applicable to our model.}

\subsection{The online model}

In the online setting, a trade can be opened at any current price, unlike the previous section where the trade opening price was always assumed to be $0$.  Moreover, a trader could decide to close a trade at any moment, not only when $w$ or $\l$ is reached.
To model this, an \emph{online trade} $T$ will be seen as a quadruple $(w, \l, p, f)$ where $w$ and $\l$ are the winning and losing prices, respectively, $p \in \mathbb{Z}$ is the opening price and $f : [\min(w, \l), \max(w, \l)] \rightarrow \mathbb{R}$ gives the profit the broker makes at any possible closing price.  We require that $f(p) = 0$, and that one of the following holds:

\begin{itemize}
    \item 
    $\l < p < w$, in which case 
    $f(q) > 0$ for each $q > p$ and $f(q) < 0$ for each $q < p$ ($up$); or
    
    \item 
    $w < p < \l$, in which case 
    $f(q) < 0$ for each $q > p$ and $f(q) > 0$ for each $q < p$ ($down$).
\end{itemize}

Note that in a real-life setting, the trader does not have to communicate the desired limit prices $w$ and $\l$.  However in our model, this would allow the trader to leave losing trades open forever.  Moreover, a limit on losses always exists, since brokers will never allow the trader's asset values to go below the account equity. 
In other words, the broker can always set a worst-case $w$ and $\l$ value if not given by the trader.

This online model can then be seen as a two-player game as follows. 
The game alternates turns between the trader and the broker, with the trader starting.  At any turn $i$, there is a current price $p_i$ and a set of open online trades $\T_i$.  On the first turn $i = 1$, the price is $p_1 = 0$ and $\T_1 = \emptyset$.

On the $i$-th turn, the trader acts first and can apply the following actions any number of times:

\begin{itemize}
    \item 
    add an online trade $T = (w, \l, p_i, f)$ to $\T_i$;
    
    \item 
    close any open trade $T = (w, \l, p, f)$ in $\T_i$.  This has the effect of removing $T$ from $\T_i$ and adding an amount of $f(p_i)$ to the broker's profit.

\end{itemize}

We say that the trader is \emph{passive} if the only action used by the trader is to add trades (never closing them unless $w$ or $\l$ is reached).
When the trader is done, the broker looks at $\T_i$ and $p_i$, and acts by choosing one of the two following actions:

\begin{itemize}
    \item 
    move the price up by one unit, so that the price for the next turn becomes $p_{i+1} = p_i + 1$;
    
    \item 
    move the price down by one unit, so that the price for the next turn becomes $p_{i+1} = p_i - 1$.

\end{itemize}

\ml{
Finally, we impose three restrictions on the broker:

\begin{enumerate}
    \item [R1.]
    if $\T_i \neq \emptyset$, the broker must eventually close a trade if the trader does nothing;
    
    \item [R2.]
    if $\T_i = \emptyset$ on the broker's turn, the broker moves towards price $0$;
    
    \item [R3.]
    the broker's decision is based solely on the current open trades $\T_i$ and the price $p_i$, and not on the past closed trades.

\end{enumerate}

These can be justified as follows.  R1 ensures constant activity, as it prevents the broker from zig-zagging in the same price area indefinitely.
R2 states that when no trade is active, the game can be reset and the current price $p_i$ might as well be interpreted as $0$.  R3 is akin to treating the broker as a Markov chain, as it makes sure that it always acts in the same manner when facing the same trade state.  Another way to justify this restriction is that the broker may assume optimal play from the trader, in which case the same move should always be preferred in a given configuration.  In real-life, traders are imperfect and brokers might profit from analyzing their past behavior, but this is beyond the scope of this work.
}

After the broker's turn, any trade $T = (w, \l, p, f)$ in $\T_i$ such that $p_{i+1} \in \{w, \l\}$ becomes closed.  In this case, $T$ is removed from $\T_i$ and a profit of $f(p_{i+1})$ is added to the broker's profit. The updated $\T_i$ becomes $\T_{i+1}$.
We also restrict $\T_i$ to be finite at any turn, meaning the trader must satisfy the requirement that there exists $t \in \mathbb{N}$ such that, regardless of the broker's strategy, for each turn $i$ we have $|\T_i| \leq t$.

It is not hard to show that if we impose no restriction on $f$, then the trader wins effortlessly.  For instance 
on turn $1$ at price $0$, the trader can simply open two trades $T_1 = (1, -1, 0, f_1)$ and $T_2 = (-1, 1, 0, f_2)$ such that $f_1(1) = f_2(-1) = -2$ and $f_1(-1) = f_2(1) = 1$.  Whichever direction the broker chooses, a profit of $-1$ will be incurred and both trades will be closed.  \ml{By R2, the trader can then wait until price $0$ is reached.  The game resets and by R3, the trader can repeat this forever, forcing a profit of $-\infty$.}

For this reason, we will require trades to be linear, as defined below.

\begin{definition}
An online trade $T = (w, \l, p, f)$ is \emph{linear} if there exists a positive rational number $\delta \in \mathbb{Q}_{> 0}$ such that:
\begin{itemize}
    \item 
    if $\l < p < w$, then $f(q) = \delta \cdot (q - p)$ for all $\l \leq q \leq w$;
    
    \item 
    if $w < p < \l$, then $f(q) = \delta \cdot (p - q)$ for all $w \leq q \leq \l$.
\end{itemize} 
\end{definition}

Note that real-life trading is done on linear trades and $\delta$ is sometimes called the \emph{lot size}.

\subsection{The only good online broker strategy}

Perhaps the first strategy that comes to mind is 
to use the results from the previous section in the online setting.  That is, on the broker's $i$-th turn, we look at $\T_i$, translate the trades so that the current price can be interpreted as $0$, and compute an optimal price trajectory $M$ using a maximum weight independent set on $G(\T_i)$.  We then go in the same direction as $M$ for one price unit, and this $M$ will be recalculated on the $(i+1)$-th turn.  Let us call this the \emph{offline strategy}.  Unfortunately, this is not a good strategy.

\begin{proposition}\label{prop:no-offline}
If the broker uses the offline strategy, then the trader can force the broker to have an infinite negative profit, even if the trader is passive and restricted to linear trades.
\end{proposition}

\begin{proof}
The trader's strategy is illustrated in Figure~\ref{fig:bad-offline}.
\begin{figure}[H]
    \centering
    \includegraphics[width=1\textwidth]{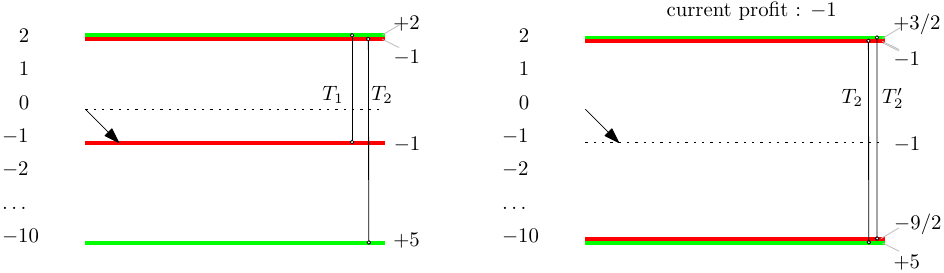}
    \caption{The offline strategy is not good.}
    \label{fig:bad-offline}
\end{figure}
On turn $1$ at price $p_1 = 0$, the trader opens a linear trade $T_1 = (2, -1, 0, f)$ with $f(-1) = -1, f(1) = 1, f(2) = 2$, and another trade $T_2 = (-10, 2, 0, g)$ such that $g(-10) = 5$ and $g(2) = -1$ (and the other values made linear with $\delta = 1/2$).
These trades are incompatible, so the broker using the offline strategy will choose to lose $T_1$ and win $T_2$, and will thus go down on its first turn.  The price will become $-1$ and the $T_1$ trade will be lost, incurring a (negative) profit of $-1$ for the \ml{broker}.   

On turn $2$ at price $p_2 = -1$, the passive trader  opens the trade $T_2' = (2, -10, -1, g')$ with $g'(-10) = -9/2$ and $g'(2) = 3/2$ (linear with $\delta = 1/2$).  At this point the trader only waits until $T_2$ and $T_2'$ are closed, \ml{which must happen by R1.}
If price $2$ is reached first, both trades close and the broker's profit for them is $-1 + 3/2 = 1/2$.  If price $-10$ is reached first instead, then the profit is $5 - 9/2 = 1/2$.  In either case, the broker's total profit is $-1 + 1/2 = -1/2$.  \ml{At this point, all trades are closed.  By R2 and R3, the trader can repeat this pattern indefinitely.}
\end{proof}

The above is actually part of a more general result that says that any strategy that does not go for the maximum \emph{potential} from the start loses infinitely.  The interest of Proposition~\ref{prop:no-offline} is that it serves as a concrete example of losing to a concrete strategy.
We next show that in fact, maximizing potential profits is essentially the only good strategy.  Anything that diverges from it is either suboptimal or has infinite negative profit.

For a set of online trades $\T$ and a price $p$, the \emph{potential profit} of $p$ on $\T$ is defined as
\[
potent(\T, p) = \sum_{(w, \l, q, f) \in \T} f(p)
\]
which represent the profit that would be made if every trade closed at price $p$.

The \emph{maximum potential} strategy is defined as follows.  On the broker's $i$-th turn, 

\vspace{3mm}

\noindent 
-- if $potent(\T_i, p_i + 1) \geq potent(\T_i, p_i - 1)$, then move the price $up$;\\
-- otherwise if $potent(\T_i, p_i + 1) < potent(\T_i, p_i - 1)$, then move the price $down$.
\vspace{3mm}

\ml{Intuitively speaking, this strategy is based on the fact that the trader either has more buys than sells or vice-versa.  Under linear trades, one of the direction has to be worse than the other (unless they both achieve profit $0$) and we simply go in the direction that is worse for the trader.  Do note that this strategy can change direction even if the trader does nothing, since winning a trade can change the direction of maximum potential.}

We now analyze this strategy against others, but from a general perspective where the broker starts applying it after a certain number of turns $i$.  

\begin{theorem}
Suppose that the trader is restricted to linear online trades.  Let $\T_i$ be a set of open trades at the start of the $i$-th turn, let $p_i$ be the current price, and let $profit(i)$ be the total profit of the broker at the start of turn $i$.  Then the following holds:

\begin{enumerate}
    \item 
    if $profit(i) + potent(\T_i, p_i) \geq 0$, then if the broker applies the maximum potential strategy from turn $i$ and onwards, it achieves a total profit of at least $profit(i) + potent(\T_i, p_i)$ against any trader.  Moreover,  this is the maximum possible profit achieved against a trader with optimal play, passive or not;
    
    \item 
    if the broker does not move the price in the direction of maximum potential profit on turn $i$, then the broker makes a profit that is strictly less than $profit(i) + potent(\T_i, p_i)$ against a trader with optimal play, passive or not;
    
    \item 
    if $profit(i) + potent(\T_i, p_i) < 0$, then the broker incurs an infinite negative profit against a trader with optimal play, passive or not.

\end{enumerate}
 
\end{theorem}

\begin{proof}
Let us first consider the first part of (1).  
Let $\T'_i$ be the set of trades after the trader has finished the $i$-th turn, and let $profit'(i)$ be the profit of the broker at this point.  Note that if the trader is passive, then $\T_i \subseteq \T'_i$, but if it is not passive, a certain set of trades might be closed.  In this case, for each $(w, \l, q, f) \in \T_i \setminus \T_i'$, a profit of $f(p_i)$ is added for the broker and an amount of $f(p_i)$ is removed from the potential of $p_i$.
Also, recall that for each new trade $(w, \l, q, f) \in \T'_i \setminus \T_i$, we have $f(p_i) = 0$.
It follows that 
\[
profit(i) + potent(\T_i, p_i) = profit'(i) + potent(\T'_i, p_i)
\]

In other words, the trader's turn does not affect the potential profit of the broker.
We next argue that one of the price directions does not decrease this potential.  
Write $\T'_i = \{T_1, \ldots, T_m\}$ and for $j \in [m]$, let $\delta_j$ be the linear factor affecting trade $T_j$.  We say that $\sign{T_j} = 1$ if $T_j$ is an $up$ trade and $\sign{T_j} = -1$ if $T_j$ is a $down$ trade.
Then by linearity, raising the price by one unit changes the potential of $T_j$ by $\sign{T_j} \cdot \delta_j$ ($up$ trades increase in potential, and $down$ trades decrease), and lowering the price by one unit changes it by $-\sign{T_j} \cdot \delta_j$.  In other words, 
\[
potent(\T'_i, p_i + 1) - potent(\T'_i, p_i) = \sum_{(w, \l, p_j, f_j) \in \T'_i} (f_j(p_i + 1) - f_j(p_i)) = \sum_{T_j \in \T'_i} \sign{T_j} \cdot \delta_j
\]
and likewise
\[
potent(\T'_i, p_i - 1) - potent(\T'_i, p_i) =  \sum_{T_j \in \T'_i} -\sign{T_j} \cdot \delta_j
\]
One of $\sum_{T_j \in \T'_i} \sign{T_j} \cdot \delta_j$ or $\sum_{T_j \in \T'_i} -\sign{T_j} \cdot \delta_j$
is greater than or equal to $0$, and so one direction leads to a potential greater than or equal to $potent(\T'_i, p_i)$.
Suppose that this direction is up.  Let $profit(i + 1)$ be the broker's profit after changing the price to $p_{i+1} = p_i + 1$ and let $\T_{i+1}$ be the set of trades still open at this point.  Some trades $(w, \l, q, f)$ might become closed at price $p_i + 1$ and incur a profit of $f(p_i + 1)$, but remove an amount of $f(p_i + 1)$ from $potent(\T'_i, p_i + 1)$.  That is, similarly as we did above, we obtain that
\[
profit(i + 1) + potent(\T_{i+1}, p_i + 1) = profit'(i) + potent(\T'_i, p_i + 1)
\]
which is greater than or equal to $profit(i) + potent(\T_i, p_i)$.  The idea is the same if going down is better for the potential instead.
By repeating this argument, this shows that the profit plus potential can be made to never decrease, guaranteeing the broker at least a profit of $profit(i) + potent(\T_i, p_i)$ when every trade is closed and the potential has reached $0$.
Also note that it is not hard to see that unless the trader opens new trades, the broker will keep the same direction until at least one trade is closed, as required by our model. This shows the first part of (1).

We now argue that this is what the optimal trader can achieve.
If the trader is not passive, then the trader can close every trade at price $p_i$ at the start of the $i$-th turn and stop playing.  The broker's profit will be exactly $profit(i) + potent(\T_i, p_i)$, which is the best the trader can hope for since this is a lower bound.  This proves (1) in the case of a non-passive trader.

We next prove a claim that will be useful for the remaining statements of the theorem that concern passive traders. 
Let $\T_j$ be any set of open trades at some current price $p_j$, and assume that it is the trader's turn. We claim that the passive trader can add a set of trades to $\T_j$ such that the broker's profit on $\T_j$ is forced to be exactly
$potent(\T_j, p_j)$.  This effectively simulates the action of closing all trades at current price.
Let $\T_j^+ = \{(w, \l, p, f) \in \T : w > p_j\}$
and let $\T_j^- = \T_j \setminus \T_j^+$.
Consider $T = (w, \l, p, f) \in \T_j^+$ and let $\delta$ be the linear factor affecting $T$.  
Add the corresponding linear trade $T' = (\l, w, p_j, f')$ such that \ml{$f'(\l) = \delta (p_j - \l)$ and $f'(w) = \delta (p_j - w)$}.
Suppose that price $q \in \{w, \l\}$ is reached first. Since the trades are linear with the same $\delta$ the profit on $T$ and $T'$ will be $f(q) + f'(q) = \delta(q - p) + \delta(p_j - q) = \delta (p_j - p)$.  Again by linearity, this is exactly $f(p_j)$.  
By symmetry, one can add a similar trade $T'$ for each $(w, \l, p, f) \in \T^-$ to force a profit of $f(p_j)$ for $T$ and $T'$.  \ml{Because of restriction R1, the trader can wait until the broker has closed every trade.}
It follows that the trader can force the current trades to close for a total profit of $\sum_{(w, \l, p, f) \in \T_j} f(p_j) = potent(\T_j, p_j)$.

This claim shows that (1) holds also for passive traders, since the trader can force a final profit of $profit(i) + potent(\T_i, p_i)$ at the start of the $i$-th turn.

Statement (2) is now easy to prove.  If the broker moves against the optimal potential, then we will get $profit(i + 1) + potent(\T_{i+1}, p_{i+1}) < profit(i) + potent(\T_i, p_i)$.
As we have seen, the trader (passive or not) can force a final profit of $profit(i + 1) + potent(\T_{i+1}, p_{i+1})$, which is strictly less than what the maximum potential strategy gives.

As for statement (3), suppose that the trader has done a sequence of $i - 1$ actions such that $profit(i) + potent(\T_i, p_i) < 0$.  Again, the trader, passive or not, can force this negative profit on the broker.  After the latter has closed every trade, the trader can simply repeat the same sequence of actions to make the broker lose the same amount again indefinitely (which works because of \ml{restrictions R2 and R3, which imply that the game always resets after this negative profit, and that the broker repeats the same mistake every time)}.
\end{proof}

As a corollary, we see that if the trader makes a move that makes the profit plus potential positive, then the broker will make profit.  Conversely, if the broker makes a mistake to make this value negative, then the trader will make it lose infinitely. 
Therefore, the only source of positive profit is due to a trader mistake, and the only source of negative profit is due to a broker mistake.  
If both players play optimally, no mistake is ever made and a profit of $0$ is made. 
However, in real life trading, each trade opened by the trader has a fee that goes into the broker's pockets.  We conclude that the best option for the trader is to not trade and try to get rich by other means.

\section{Conclusion}

In this work, we have provided some useful algorithmic tools for malevolent brokers to plot against traders.  
Our models are combinatorial and do not consider any form of stochasticity.  This may serve as criticism towards our model, but one must reckon that stochastic trading models have been studied for a very long time.  
Our work offers a new perspective on trade analysis and shows how price movements can be optimized when given full control over them.  Let us also mention that our graph theoretical view on trading is novel and allowed us to discover a potentially new graph class (although they might just be chordal bipartite graphs --- the future will tell).  

One future direction would be to reach a middle ground between combinatorial and stochastic by incorporating various forms of randomness into our price manipulation possibilities.  For instance, we may be able to move prices in the desired direction only with certain probabilities.  Also, in the online model, we often assumed a trader with optimal play, whereas a randomized trader could be considered more realistic.  
Finally, one can ask whether a trader knowledgeable of the broker's scheme is able to profit from that information.  In the offline setting, if the trader was able to see every trade, he or she could predict the broker's movements and add appropriate winning trades.  This would need to be done so that the broker would not alter its trajectory because of the new trades, leading to another algorithmic problem of interest.

\bibliographystyle{plainurl}

\bibliography{forex}

\end{document}